\documentclass[a4paper]{article}

\usepackage{RR}

\usepackage{graphicx,subfigure} 
\usepackage{algorithm}
\usepackage[isolatin]{inputenc}
\usepackage[OT1]{fontenc}                       
\usepackage{latexsym}
\usepackage{epsfig}
\usepackage{fancybox}

\newenvironment{proof}{{\bf Proof. } }{{\hfill $\Box$}\vspace{.5pc}}
\floatstyle{ruled}
\newfloat{Algorithm}{thp}{lop}[section]
\newcommand{\qed}{\hfill$\Box$}

\newtheorem{theorem}{Theorem}

\newtheorem{lemma}{Lemma}

\newtheorem{corollary}{Corollary}

\RRNo{6880}

\RRdate{Mars 2009}

\RRtitle{L'Auto-stabilisation d'un Ensemble Maximal Indépendant dans les Réseaux Unidirectionels est Difficile}
\RRetitle{Stabilizing Maximal Independent Set in Unidirectional Networks is Hard}

\titlehead{Stabilizing Maximal Independent Set in Unidirectional Networks is Hard}

\RRauthor{Toshimitsu Masuzawa\thanks{Osaka University, Japan}
\and
Sébastien Tixeuil\thanks{Université Pierre \& Marie Curie - Paris 6, LIP6-CNRS \& INRIA Grand Large, France}}

\RRresume{Un algorithme distribué est auto-stabilisant si après que des fautes et des attaques ont frappé le système et l'ont placé dans un état global arbitraire, le système récupère en temps fini un fonctionnement correct sans intervention extérieure. Dans cet article, nous considérons le problème de la construction auto-stabilisante d'un \emph{ensemble maximal indépendant} dans des réseaux uniformes et unidirectionels quelconques. Nous présentons un résultat négatif qui indique que dans les réseaux uniformes, l'auto-stabilisation \emph{déterministe} de ce problème est \emph{impossible} à résoudre. De plus, la propriété de \emph{silence} (\emph{i.e.} garantir qu'à partir d'un point de chaque exécution, les communications entre les n\oe uds du réseau sont fixées) est impossible à garantir, tant pour les variantes deterministes que probabilistes des protocoles.

Nos résultats positifs sont multiples. Nous présentons un protocole déterministe pour les réseaux unidirectionels \emph{identifiés} quelconques qui présente une complexité en temps et en espace qui reste polynomiale, avec un ordonnancement asynchrone. Nous complétons l'étude avec des protocoles probabilistes dans le cas uniforme~: le premier protocole requiert une mémoire infinie mais supporte un ordonnancement asynchrone, le deuxième protocole utilise une mémoire polynomiale mais requiert un ordonnancement synchrone. Les deux protocoles ont une compexité moyenne en temps polynomiale.}

\RRabstract{A distributed algorithm is self-stabilizing if after faults and attacks hit the system and place it in some arbitrary global state, the system recovers from this catastrophic situation without external intervention in finite time. In this paper, we consider the problem of constructing self-stabilizingly a \emph{maximal independent set} in uniform unidirectional networks of arbitrary shape. On the negative side, we present evidence that in uniform networks, \emph{deterministic} self-stabilization of this problem is \emph{impossible}. Also, the \emph{silence} property (\emph{i.e.} having communication fixed from some point in every execution) is impossible to guarantee, either for deterministic or for probabilistic variants of protocols.

On the positive side, we present a deterministic protocol for networks with arbitrary unidirectional networks with unique identifiers that exhibits polynomial space and time complexity in asynchronous scheduling. We complement the study with probabilistic protocols for the uniform case: the first probabilistic protocol requires infinite memory but copes with asynchronous scheduling, while the second probabilistic protocol has polynomial space complexity but can only handle synchronous scheduling. Both probabilistic solutions have expected polynomial time complexity.}

\RRmotcle{Syst\`{e}mes distribu\'{e}s, Algorithme distribu\'{e}, Ensemble Maximal Indépendant, Réseaux Unidirectionels, Auto-stabilisation, Auto-stabilisation probabiliste}
\RRkeyword{Distributed systems, Distributed algorithm, Maximal Independent Set, Unidirectional Networks, Self-stabilization, Probabilistic self-stabilization}

\RRprojet{Grand large}
\RRtheme{\THNum}
\URFuturs

\begin{document}

\makeRR
\section{Introduction}

One of the most versatile technique to ensure forward recovery of distributed systems is that of \emph{self-stabilization}~\cite{D74j,D00b}. A distributed algorithm is self-stabilizing if after faults and attacks hit the system and place it in some arbitrary global state, the system recovers from this catastrophic situation without external (\emph{e.g.} human) intervention in finite time. 

The vast majority of self-stabilizing solutions in the literature~\cite{D00b} considers \emph{bidirectional} communications capabilities, \emph{i.e.} if a process $u$ is able to send information to another process $v$, then $v$ is always able to send information back to $u$. This assumption is valid in many cases, but can not capture the fact that asymmetric situations may occur, \emph{e.g.} in wireless networks, it is possible that $u$ is able to send information to $v$ yet $v$ can not send any information back to $u$ ($u$ may have a wider range antenna than $v$). Asymmetric situations, that we denote in the following under the term of \emph{unidirectional} networks, preclude many common techniques in self-stabilization from being used, such as preserving local predicates (a process $u$ may take an action that violates a predicate involving its outgoing neighbors without $u$ knowing it, since $u$ can not get any input from its outgoing neighbors).

\paragraph{Related works}

Self-stabilizing solutions are considered easier to implement in bidirectional networks since detecting incorrect situations requires less memory and computing power~\cite{BDDT07j}, recovering can be done locally~\cite{AD02j}, and Byzantine containment can be guaranteed~\cite{MT06cb,MT07j,NA02c}.

Investigating the possibility of self-stabilization in unidirectional networks was recently emphasized in several papers~\cite{AB98j,CG01c,DDT99j,DDT06j,DS04j,DT01jb,DT02j,DT03j,BDGT09c}. However, topology or knowledge about the system varies:~\cite{DDT99j} considers \emph{acyclic} unidirectional networks, where erroneous initial information may not loop; \cite{AB98j,CG01c,DT02j,DS04j} assume \emph{unique identifiers} and \emph{strongly connected} so that global communication can be implemented; ~\cite{DDT06j,DT01jb,DT03j} makes use of \emph{distinguished processes} yet operate on arbitrary unidirectional networks. 

Tackling arbitrary \emph{uniform} unidirectional networks in the context of self-stabilization proved to be hard. In particular, \cite{BDGT09c,BDGT08r} studied the self-stabilizing vertex coloring problem in unidirectional uniform networks (where adjacent nodes must ultimately output different colors). Deterministic and probabilistic solutions to the vertex coloring problem~\cite{GT00c,MFGST06ca} in bidirectional networks have \emph{local} complexity ($\Delta$ states per process are required, and $O(\Delta)$ --resp. $O(1)$-- actions per process are needed to recover from arbitrary state in the case of a deterministic --resp. probabilistic-- algorithm). By contrast, in unidirectional networks, \cite{BDGT09c} proves a lower bound of $n$ states per process (where $n$ is the network size) and a recovery time of at least $n(n-1)/2$ actions in total (and thus $\Omega(n)$ actions per process) in the case of deterministic uniform algorithms, while \cite{BDGT08r} provides a probabilistic solution that remains either local in space \emph{or} local in time, but not both.

\paragraph{Our contribution}

In this paper, we consider the problem of constructing self-stabilizingly a \emph{maximal independent set} in uniform unidirectional networks of arbitrary shape. It turns out that local maximization (\emph{i.e.} maximal independent set) is strictly more difficult than local predicate maintainance (\emph{i.e.} vertex coloring). On the negative side, we present evidence that in uniform networks, \emph{deterministic} self-stabilization of this problem is \emph{impossible}. Also, the \emph{silence} property (\emph{i.e.} having communication fixed from some point in every execution) is impossible to guarantee, either for deterministic or for probabilistic variants of protocols.

On the positive side, we present a deterministic protocol for networks with arbitrary unidirectional networks with unique identifiers that exhibits $O(m \log n)$ space complexity and $O(D)$ time complexity in asynchronous scheduling, where $n$ is the network size and $D$ is the network diameter. We complement the study with probabilistic protocols for the uniform case: the first probabilistic protocol requires infinite memory but copes with asynchronous scheduling (stabilizing in time $O(\log n + \log \ell + D)$, where $\ell$ denotes the number of fake identifiers in the initial configuration), while the second probabilistic protocol has polynomial space complexity (in $O(m \log n)$) but can only handle synchronous scheduling (stabilizing in time $O((n+\ell)\log n)$).

\paragraph{Outline}

The remaining of the paper is organized as follows: Section~\ref{sec:model} presents the programming model and problem specification. Section~\ref{sec:impossible} presents our negative results, while Section~\ref{sec:possible} details the protocols. Section~\ref{sec:conclusion} gives some concluding remarks and open questions.

\section{Preliminaries}
\label{sec:model}

\paragraph{Program model} 

A program consists of a set $V$ of $n$ processes. A process maintains a set of variables that it can read or update, that define its \emph{state}. 
A process contains a set of \emph{constants} that it can read but not update. A binary relation $E$ is defined over distinct processes such that $(i,j) \in E$ if and only if $j$ can read the variables maintained by $i$; $i$ is a \emph{predecessor} of $j$, and $j$ is a \emph{successor} of $i$. The set of predecessors (resp. successors) of $i$ is denoted by $P.i$ (resp. $S.i$), and the union of predecessors and successors of $i$ is denoted by $N.i$, the \emph{neighbors} of $i$. The \emph{ancestors} of process $i$ is recursively defined as follows:  predecessors of $i$ are ancestors of $i$, and ancestors of each predecessor of $i$ are also ancestors of $i$.  The \emph{descendants} of $i$ are similarly defined using successors (instead of predecessors).

For processes $i$ and $j$ in $V$, $d(i,j)$ denotes the \emph{distance} (or the length of the shortest path) \emph{from} $i$ \emph{to} $j$ in the directed graph $(V,E)$. We define, for convenience, the distance as $d(i,i)=0$ and $d(i,j)=\infty$ if $i$ is not reachable to $j$. The \emph{diameter} $D$ is defined as $D=\max \{d(i,j)~ |~ (i,j) \in V \times V, d(i,j) \ne \infty\}$.

An action has the form $\langle name \rangle : \langle guard \rangle \longrightarrow \langle command \rangle$. A \emph{guard} is a Boolean predicate over the variables of the process and its predecessors. A \emph{command} is a sequence of statements assigning new values to the variables of the process. We refer to a variable $v$ and an action $a$ of process $i$ as $v.i$ and $a.i$ respectively. A \emph{parameter} is used to define a set of actions as one parameterized action.  
Notice that actions of a process are completely independent of its successors.

A \emph{configuration} of the program is the assignment of a value to every variable of each process from its corresponding domain. Each process contains a set of actions. 
In some configuration, an action is \emph{enabled} if its guard is \textbf{true} in the configuration, and a process is \emph{enabled} if it has at least one enabled action in the configuration.  A \emph{computation} is a maximal sequence of configurations $\gamma_0, \gamma_1, \ldots$ such that for each configuration $\gamma_i$, the next configuration $\gamma_{i+1}$ is obtained by executing the command of at least one action that is enabled in $\gamma_i$.  Maximality of a computation means that the computation is infinite or it terminates in a configuration where none of the actions are enabled. A program that only has terminating computations is \emph{silent}.

A \emph{scheduler} is a predicate on computations, that is, a scheduler is a set of possible computations, such that every computation in this set satisfies the scheduler predicate.  
We consider only \emph{weakly fair} schedulers, where no process can remain enabled in a computation without executing any action. 
We distinguish three particular schedulers in the sequel of the paper: the \emph{distributed} scheduler corresponds to predicate \textbf{true} (that is, all weakly fair computations are allowed).  The \emph{locally central} scheduler implies that in any configuration belonging to a computation satisfying the scheduler, no two enabled actions are executed simultaneously on neighboring processes.  The \emph{synchronous} scheduler implies that in any configuration belonging to a computation satisfying the scheduler, every enabled process executes one of its enabled actions.

The distributed and locally central schedulers model \emph{asynchronous} distributed systems.  In asynchronous distributed systems, time is usually measured by \emph{asynchronous rounds} (simply called \emph{rounds}).
Let $E= \gamma_0,\gamma_1, \ldots$ be a computation. The first round of $E$ is the minimum prefix of $E$, $E_1=\gamma_0,\gamma_1,\ldots,\gamma_k$, such that every enabled process in $\gamma_0$ executes its action or becomes disabled in $E_1$.  Round $t\ (t\ge 2)$ is defined recursively, by applying the above definition of the first round to $E'=\gamma_k,\gamma_{k+1},\ldots$. Intuitively, every process has a chance to update its state in every round.

A configuration \emph{conforms} to a predicate if this predicate is
\textbf{true} in this configuration; otherwise the configuration \emph{violates} the
predicate. By this definition every configuration conforms to predicate
\textbf{true} and none conforms to \textbf{false}. Let $R$ and $S$ be
predicates over the configurations of the program. Predicate $R$ is
\emph{closed} with respect to the program actions if every configuration of the
computation that starts in a configuration conforming to $R$ also conforms to
$R$. Predicate $R$ \emph{converges} to $S$ if $R$ and $S$ are closed and any
computation starting from a configuration conforming to $R$ contains a
configuration conforming to $S$. The program \emph{deterministically
stabilizes} to $R$ if and only if \textbf{true} converges to $R$. The program
\emph{probabilistically stabilizes} to $R$ if and only if \textbf{true}
converges to $R$ with probability $1$.

\paragraph{Problem specification} 

Each process $i$ defines a function $\mathit{mis}.i$ that takes as input the states of $i$ and its predecessors, and outputs a value in $\{\mathbf{true}, \mathbf{false}\}$. The \emph{unidirectional maximal independent set} (denoted by UMIS in the sequel) predicate is satisfied if and only if for every $i\in V$, either $\mathit{mis}.i = \mathbf{true} \wedge \forall j\in N.i, \mathit{mis}.j = \mathbf{false}$ or $\mathit{mis}.i = \mathbf{false} \wedge \exists j\in N.i, \mathit{mis}.j = \mathbf{true}$.

\section{Impossibility Results in anonymous networks}
\label{sec:impossible}

In this section, we consider anonymous and uniform networks, where processes of the same in-degree execute exactly the same code (note however that probabilistic protocols may exhibit different actual behaviors when making use of a random variable).

\begin{theorem}
\label{th:imp-silent}
There exists no silent self-stabilizing solution for the UMIS problem.
\end{theorem}

\begin{proof}
Assume there exists such a solution and consider System $A$ as depicted in Figure~\ref{fig:silent}.$(a)$. Since the protocol is silent, it reaches a terminal configuration where exactly one of the three processes, says $a$, has $\mathit{mis}.a = \mathbf{true}$. Now consider the system in Figure~\ref{fig:silent}.$(b)$, with the states of the processes in the tail (that is $b'$ and $c'$) being the same as those in the $3$-cycle (that is $b$ and $c$). Both processes with state $S'$ ($b$ and $b'$) have the same in-degree and the same predecessor; as the one in the cycle ($b$) is silent, the second one ($b'$) is also silent. Both processes with state $S''$ ($c$ ad $c'$) have the same in-degree and the same predecessor state; as the one in the cycle ($c$) is silent, the second one ($c'$) is also silent. As a result, both processes $b'$ and $c'$ in the tail of System $B$ never move. Since the UMIS function is based solely on the current state, in-degree, and predecessor state, the UMIS function returns the same result for both processes $b$ and $b'$ in state $S'$ and for both processes $c$ and $c'$ in state $S''$. So, both processes $b'$ and $c'$ in the tail are not in the UMIS. Overall, System $B$ describes a terminal configuration that is not a maximal independent set (the UMIS predicate does not hold at $c'$).
\end{proof}

\begin{figure}

\centering

\subfigure[System $A$]{\includegraphics[height=3.5cm]{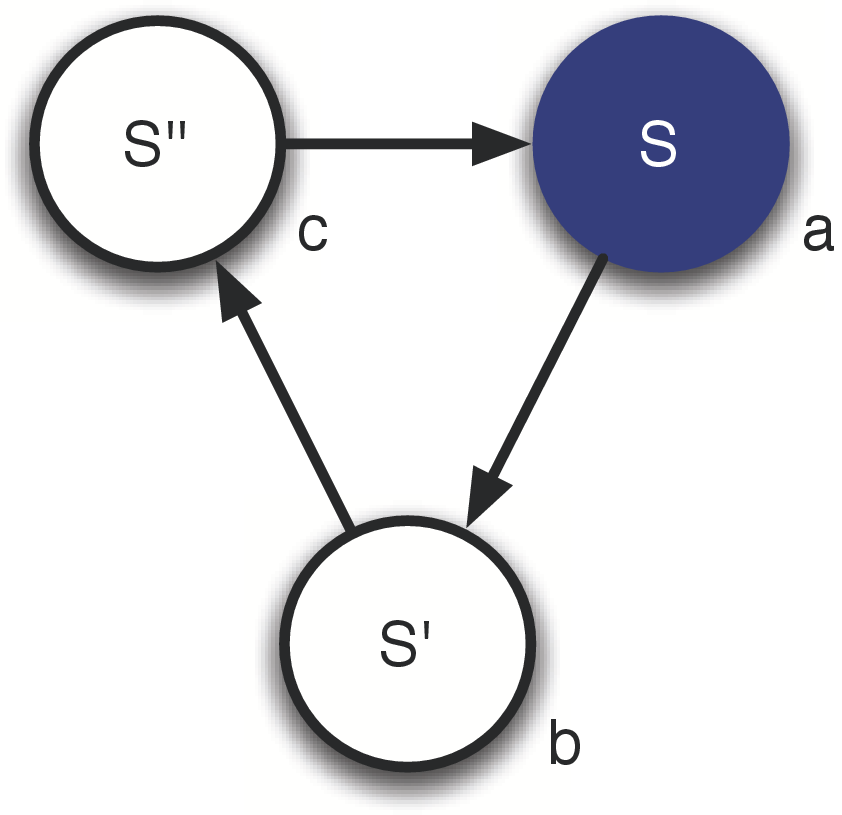}}
\subfigure[System $B$]{\includegraphics[height=3.5cm]{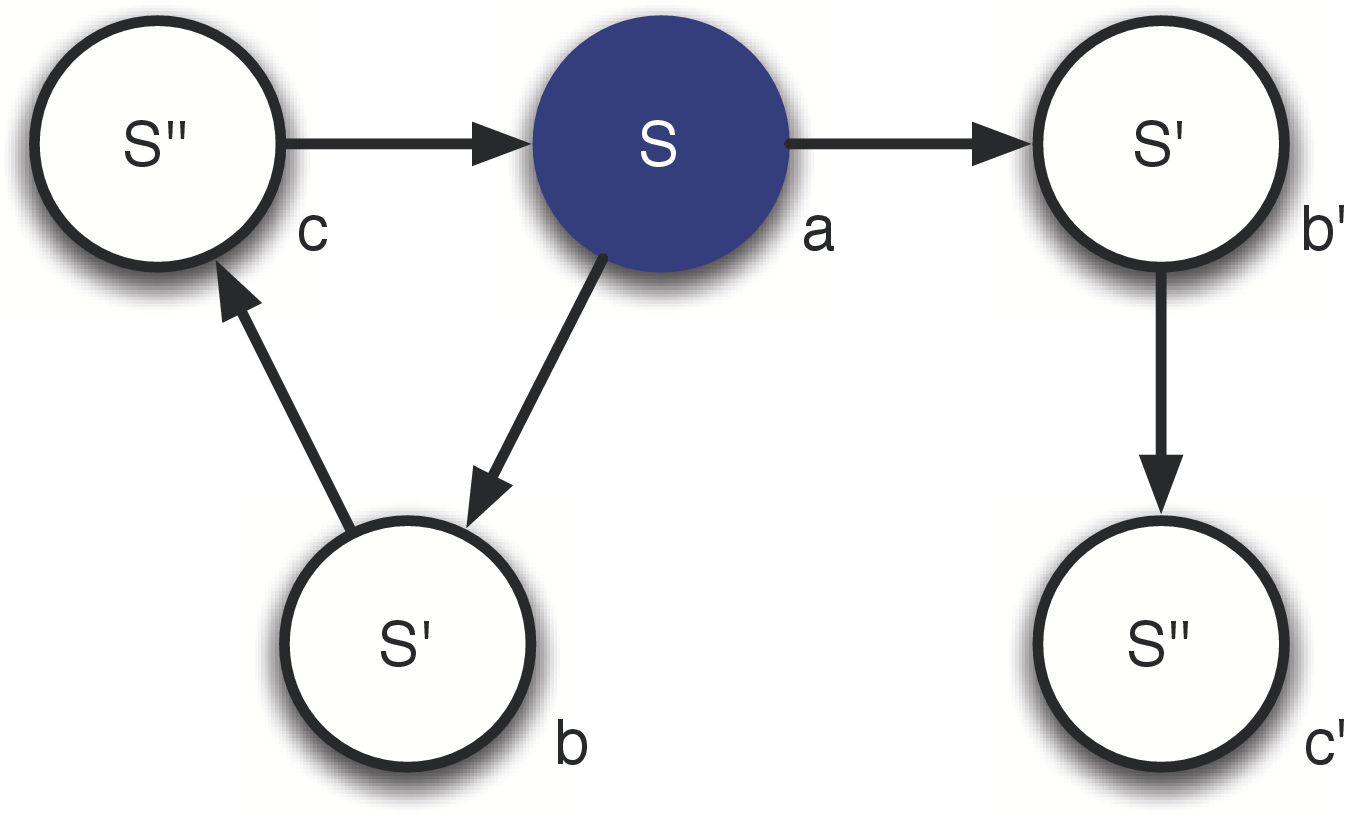}}

\caption{Impossibility of self-stabilizing UMIS}
\label{fig:silent}

\end{figure}

Notice that the impossibility results of Theorem \ref{th:imp-silent} holds even for probabilistic potential solutions. We now prove that relaxing the silence property still prevents the existence of deterministic solutions.

\begin{theorem}
There exists no deterministic self-stabilizing solution for the UMIS problem.
\end{theorem}

\begin{proof}
Assume there exists such a solution and consider the two systems $A$ and $B$ that are depicted in Figure~\ref{fig:silent}.
We consider a computation of system $A$, that eventually ends up in a stable output of the $\mathit{mis}$ function for all processes $a$, $b$, and $c$ ($a$ being the one process with $\mathit{mis}.a=\mathbf{true}$), and construct a sibling execution in System $B$ as follows:
\begin{itemize}
\item processes $b$ and $b'$ (resp. $c$ and $c'$) in System $B$ have the same 
initial states as $b$ (resp. $c$) in System $A$,
\item anytime process $b$ (resp. $c$) is executed in System $A$, both processes $b$ and $b'$ (resp. $c$ and $c'$) are executed in System $B$, 
\item anytime $a$ is executed in System $A$, $a$ is also executed in System $B$.
\end{itemize}

Now, at any time, in System $B$, both processes $b$ and $b'$ are in the same state, with the same predecessors' states. As a result, the output of their $\mathit{mis}$ function is the same. The same holds for processes $c$ and $c'$. Since System $A$ eventually ends up in a configuration from which all $\mathit{mis}$ functions are stable, the same holds for system $B$, where $\mathit{mis}.b'$ and $\mathit{mis}.c'$ both return $\mathbf{false}$. As a result, a UMIS is never constructed in System $B$.  
\end{proof}

%
%
%
%
%
%
%
%
%

\section{Possibility Results}
\label{sec:possible}

The previous impossibility results yield that for the deterministic case, only non uniform networks admit a self-stabilizing solution for the UMIS problem. In section~\ref{sec:deter}, we present such a deterministic solution.

For anonymous and uniform networks, there remains the probabilistic case. We proved that probabilistic yet silent solutions are impossible, so both our solutions are non-silent. The one that is presented in Section~\ref{sec:proba1} performs in asynchronous networks but requires unbounded memory, while the one that is presented in Section~\ref{sec:proba2} performs in synchronous networks and uses $O(m \log n)$ memory per process. 

\subsection{Deterministic solution with identifiers}
\label{sec:deter}

The intuition of the solution is as follows. Every process collects the predecessor information from all of its ancestors using the self-stabilizing approach given in~\cite{DT02j,DH97j,DT03j}.  From the collected information, each process $i$ can reconstruct the exact topology of the subgraph consisting of all the ancestors and $i$ itself. Then, depending on where the process is located, two possibilities can be considered:

\begin{enumerate}
\item The process is in a strongly connected component that includes all of its ancestors. In the directed acyclic graph of strongly connected components, this process is in a \emph{source} component. Then every process in the source component constructs the same topology. The MIS in this source component is constructed for example by giving processes priority in the descending order of identifiers (\emph{i.e.}, the process with maximal identifier has highest priority).
\item The process is in a \emph{non-source} strongly connected component in the same acyclic graph of strongly connected components. 
Then, the same process as in the previous situation repeats, with the additional constraint that stronger priority is given to the processes in the upwards strongly connected components.
\end{enumerate}  

The detailed algorithm is given in Algorithm \ref{alg:det}.

\begin{Algorithm}[htb]
\begin{tabbing}
xxx \= xxx \= xxx \= xxx \= xxx \= \kill
{\tt constants of process} $i$ \\
\> $id_i$: identifier of $i$; \\
\> $P_i$: identifier set of $P.i$; \\
{\tt variables of process} $i$ \\
\> {\it Topology}$_i$: set of $(id, ID, dist)$ tuples; 
   // topology that $i$ is currently aware of. \\
\> \> // $id$: a process identifier \\
\> \> // $ID$: identifier set of $P.(id)$ \\
\> \> // $dist$: distance from $id$ to $i$ in {\it Topology}$_i$. \\
{\tt function} \\
\> {\it update}({\it Topology}$_i$)\\
\> \> {\it Topology}$_i$ := $\{(id_i, P_i, 0)\} \cup 
      \bigcup_{j \in P.i}\{(id,ID,dist+1)~|~(id,ID,dist) \in 
      ${\it Topology}$_j\}$; \\
\> \> {\tt while} $\exists (id, ID, dist), (id', ID', dist') \in$ 
      {\it Topology}$_i$ s.t. $id = id'$ {\tt and} $dist < dist' $ \\
\> \> \> remove $(id', ID', dist')$ from {\it Topology}$_i$; \\
\> \> {\tt while} $\exists (id, ID, dist), (id', ID', dist') \in$ 
      {\it Topology}$_i$ s.t. $id = id'$ {\tt and} $ID \ne ID'$ \\
\> \> \> remove one of them (arbitrarily) from {\it Topology}$_i$; \\
\> \> {\tt while} $\exists (id, ID, dist) \in$ {\it Topology}$_i$ s.t. 
      $id$ is unreachable to $i$ in {\it Topology}$_i$ \\
\> \> \> remove $(id, ID, dist)$ from {\it Topology}$_i$; \\
\> {\it UMIS}$(${\it Topology}$_i)$\\
\> \> {\it WorkingTp}$_i$ := {\it Topology}$_i$; \\
\> \> {\it UMIS}$_i := \emptyset$ \\
\> \> {\tt while} $\exists (id_i, P_i, 0) \in$ {\it WorkingTp}$_i$ \{ \\
\> \> \> Let $W$ be a source strongly connected component of {\it WorkingTp}$_i$; \\
\> \> \> {\tt for each} $id \in W$ in the descending order of identifiers \\
\> \> \> \> {\tt if} {\it UMIS}$_i \cup \{id\}$ is an independent set \\
\> \> \> \> \> {\it UMIS}$_i :=$ {\it UMIS}$_i \cup \{id\};$ \\
\> \> \> {\it WorkingTp}$_i$ := {\it WorkingTp}$_i - W$; \\
\> \> \} \\
\> \> {\tt if} $id_i \in$ {\it UMIS}$_i$ \\
\> \> \> {\tt output} {\tt true}; \\
\> \> {\tt else} \\
\> \> \> {\tt output} {\tt false}; \\
{\tt actions of process} $i$ \\
\> ${\tt true} \longrightarrow$ {\it update}({\it Topology}$_i$);
   {\it UMIS}({\it Topology}$_i);$ 
\end{tabbing}
\caption{Deterministic UMIS algorithm in asynchronous networks with identifiers}
\label{alg:det}
\end{Algorithm}

\begin{lemma}
\label{lem:top-det}
Let $i$ be any process.  At the end of the $k$-th round $(k \ge 1)$ and later, the topology stored in variable {\it Topology}$_i$ is correct up to distance $k-1$:
\begin{enumerate}
\item
for every process $j$ with $d(j,i) \le k-1$, {\it Topology}$_i$ stores 
the correct tuple $(j, P.j, d(j,i))$ of $j$, and 
\item
every tuple $(id, ID, d) \in$ {\it Topology}$_i$ is the correct one
$(j, P.j, d(j,i))$ of some process $j$ if $d \le k-1$.
\end{enumerate}
\end{lemma}

\begin{proof}
We prove the lemma by induction on $k$.
Let us observe first that the lemma holds for $k=1$ (inductive basis): Once $i$ executes its action, {\it Topology}$_i$ always contains $(i, P.i, 0)$ and any other tuple $(id, ID, d)$ satisfies $d \ge 1$.

Assuming that the lemma holds for $k$ (inductive hypothesis),
we now prove the lemma for $k+1$ (inductive step).  
Any process $u$ with $d(u, i) \le k$ satisfies $d(u, j) = d(u, i)-1 \le k-1$ for 
some predecessor $j$ of $i$.
From the inductive hypothesis, {\it Topology}$_j$ contains the correct tuple
$(u, P.u, d(u,j))$ of $u$ at the end of the $k$-th round and later.
Thus, $i$ reads the correct tuple $(u, P.u, d(u,j))$ in {\it Topology}$_j$
and updates its distance correctly at every action in the $(k+1)$-th round 
and later.
The hypothesis also implies that any tuple $(u, ID, d)$ 
contained in {\it Topology}$_{v}$ of any predecessor $v$ of $i$ after the end of 
the $k$-th round satisfies $d \ge d(u, i)-1$ and is correct if $d = d(u, i)-1$. 
Thus, the correct tuple $(u, P.u, d(u,i))$ is never removed from 
{\it Topology}$_i$ in the $(k+1)$-th round or later.  
The first claim of the lemma holds for $k+1$.

Existence of tuple $(id, ID, d)\ (d \ne 0)$ in {\it Topology}$_{i}$ at 
the end of the $(k+1)$-th round or later implies that $i$ reads $(id, ID, d-1)$ 
in {\it Topology}$_{j}$ of some predecessor $j$ of $i$.  From the hypothesis, 
any tuple $(id, ID, d-1)$ contained in {\it Topology}$_{j}$ after the end of 
the $k$-th round is \emph{correct} (or $id$ is an identifier of a really 
existing process, say $v$, $ID$ is the identifier set of $P.v$ and $d=d(v,j)$
holds) if $d-1 \le k-1$.
Thus, any tuple $(id, ID, d)$ contained in {\it Topology}$_{i}$ at the end of 
the $(k+1)$-th round or later is correct if $d \le k$. 
The second claim of the lemma holds for $k+1$.
\end{proof}

The following corollary is derived from Lemma \ref{lem:top-det}.

\begin{corollary}
\label{col:top-det}
Let $i$ be any process and $D(i)$ be the maximum distance
to $i$ from all the ancestors of $i$.  At the end of the $(D(i)+1)$-th round and later, {\it Topology}$_i$ stores the exact topology of the subgraph consisting of all the ancestors of $i$ and $i$ itself.
\end{corollary}

\begin{proof}
Concerning {\it Topology}$_i$ at the end of the $(D(i)+1)$-th round and later, Lemma \ref{lem:top-det} shows that the correct tuple $(u, P_u, d(u,i))$ of every ancestor $u$ of $i$ is contained, and any tuple $(id, ID, d)$ with $d \le D(i)$ is correct.
This implies that {\it Topology}$_i$ at the end of the $(D(i)+1)$-th round and later can contain no tuple $(id, ID, d)$ with $d > D(i)$ since the process with identifier $id$ is not reachable to $i$ in {\it Topology}$_i$ and such a tuple is removed from {\it Topology}$_i$ if exists.  Thus the corollary holds.
\end{proof}

\begin{theorem}
\label{th:top-det}
Algorithm \ref{alg:det} presents a self-stabilizing deterministic UMIS algorithm
in asynchronous networks with identifiers.  Its convergence time is $D+1$
rounds where $D$ is the diameter of the network, and the memory space used at each node is $O(m \log n)$ bits.
\end{theorem}

\begin{proof}
Let $Topology$ be the exact topology of the network.
It is obvious that {\it UMIS}({\it Topology}) correctly finds a UMIS when executed until ${\it WorkingTP} = \emptyset$ holds.  When {\it Topology}$_i$ stores the exact topology of the subgraph consisting of all ancestors of $i$, {\it UMIS}({\it Topology}$_i)$ selects $i$ as a member of UMIS iff {\it UMIS}({\it Topology}) selects $i$: whether process $i$ is selected by {\it UMIS}({\it Topology}) depends only on the topology of the subgraph consisting of all ancestors of $i$.
Corollary \ref{col:top-det} guarantees that {\it Topology}$_i$ of every process $i$ stores such exact topology at the end of the $(D+1)$-th round and later, and thus, the theorem holds. As the {\it Topology} variable may end up in containing an entry for every node, the over space needed is $O(m \log n)$ bits per process.
\end{proof}

Notice that Algorithm \ref{alg:det} enables each process $i$ to know eventually the exact topology of the subgraph consisting of all the ancestors of $i$.  Algorithm \ref{alg:det} can be easily extended so that each process can eventually get the exact topology containing the input values of the ancestors if each process has a \emph{static} input value. Such an extension results in a \emph{universal} scheme since it can solve any non-reactive problem that is consistently solvable at each process using the topology and the input values of its ancestors.

Another observation is that Algorithm \ref{alg:det} can easily be modified to become \emph{silent}. For simplicity of our presentation, every process always has an enabled action with guard {\bf true}, and thus, Algorithm \ref{alg:det} is not silent.  But, Algorithm \ref{alg:det} becomes silent by changing the guard so that the action becomes enabled only when {\it Topology}$_i$ needs to be updated.

\subsection{Probabilistic solution with unbounded memory in asynchronous anonymous networks}
\label{sec:proba1}

In this subsection, we present a probabilistic self-stabilizing UMIS algorithm
for asynchronous \emph{anonymous} networks.
The solution is based on a probabilistic unique naming of the processes and a deterministic UMIS algorithm that assumes unique process identifiers. 
In the naming algorithm, each process is given a name variable that can be arbitrary large (thus the unbounded memory requirement). The naming is unique with probability $1$ after a bounded number of new name draws. The new name draw consists in appending a random bit at the end of the current identifier. Each time the process is activated, a new random bit is appended. 
In parallel, we essentially run the deterministic UMIS algorithm presented in the previous subsection.  
The main difference from the previous algorithm is in handling the process identifiers. The variable {\it Topology} of a particular process may contain several different identifiers of a same process since the identifier of the process continues to get longer and longer in every execution of the protocol. To circumvent the problem, we consider two distinct identifiers to be the \emph{same} if one is a prefix of the other, and anytime such same identifiers conflict, only the longest one is retained.
Another difference is that we do not need the distance information.  The distance information is used in the previous algorithm to remove the \emph{fake} tuples $(i, ID, d)$ of process $i$ such that $ID \ne P.i$, which may exist in the initial configuration. In our scheme, tuples with fake identifiers that are prefixes of identifiers of real processes are eventually removed in Algorithm \ref{alg:proba1} since the correct identifier eventually becomes longer than any fake identifier. Other tuples with fake identifiers are eventually disconnected from the constructed subgraph topology.



The details of the algorithm are given in Algorithm \ref{alg:proba1}; only the topology update part is described since the UMIS function is the same as in Algorithm \ref{alg:det}.

\begin{Algorithm}[htb]
\begin{tabbing}
xxx \= xxx \= xxx \= xxx \= xxx \= \kill
{\tt variables of process} $i$ \\
\> $id_i$: identifier (binary string) of $i$; \\
\> $P_i$: identifier set of $P.i$; \\
\> {\it Topology}$_i$: set of $(id, ID)$ tuples; 
   // topology that $i$ is currently aware of. \\
\> \> // $id$: a process identifier \\
\> \> // $ID$: identifier set of $P.(id)$ \\
{\tt function} \\
\> {\it update(Topology}$_i$)\\
\> \> $id_i :=$ {\it append}($id_i$, {\it random\_bit}); 
      // append a random bit to the current id \\
\> \> {\it Topology}$_i$ := $\{(id_i, P_i)\} \cup 
      \bigcup_{j \in P.i}${\it Topology}$_j$; \\
\> \> {\tt while} $\exists (id, ID), (id', ID') \in$ 
      {\it Topology}$_i$ s.t. $id'$ is a prefix of $id$\\
\> \> \> remove $(id', ID')$ from {\it Topology}$_i$; \\
\> \> {\tt while} $\exists (id, ID) \in$ {\it Topology}$_i$ s.t. 
      $id$ is unreachable to $i$ in {\it Topology}$_i$ \\
\> \> \> remove $(id, ID)$ from {\it Topology}$_i$; 
\end{tabbing}
\caption{Probabilistic UMIS algorithm in asynchronous anonymous networks}
\label{alg:proba1}
\end{Algorithm}

\begin{theorem}
\label{th:proba1}
Algorithm \ref{alg:proba1} presents a self-stabilizing probabilistic UMIS algorithm in asynchronous anonymous networks.  Its expected convergence time is $O(\log n + \log \ell + D)$ rounds where $D$ is the diameter of the network and $\ell$ is the number of fake identifiers in the initial configuration.
\end{theorem}

\noindent
{\bf Proof Sketch:} 
It is clear that the identifier of any process eventually becomes distinct from any other's with probability 1.  We first show that every process has a unique identifier in $O(\log n)$ expected rounds.

We consider, as the worst-case scenario, the case where all processes start with the same identifier and each process appends only a single bit to its identifier in every round.
%

The probability that every process has a unique identifier at the end of round $k$ (\emph{i.e.}, $n$ random strings of $k$ bits are mutually distinct) is evaluated as follows when $n$ is small compared to $2^k$:
\[\prod_{i=1}^{n-1}(1-\frac{i}{2^k}) \approx 
  \prod_{i=1}^{n-1}exp(-\frac{i}{2^k}) = exp(-\frac{n(n-1)}{2^{k+1}}) \approx
  exp(-\frac{n^2}{2^{k+1}})\]
We introduce a discrete random variable $X$ to represent the number of rounds required until every process has a unique identifier.  When we consider the execution after round $2 \log n$ to guarantee $n$ is small compared to $2^k$,
the expected number of rounds is then bounded by
\[2 \log n + \sum_{i=2 \log n}^{\infty}Prob(X > i) = 
  2 \log n + \sum_{i=2 \log n}^{\infty}(1- exp(-\frac{n^2}{2^{i+1}}))\]
\[\approx 2 \log n + \sum_{i=2 \log n}^{\infty}\frac{n^2}{2^{i+1}} =
  2 \log n + O(1) \]
Thus, every process has a unique identifier in expected $O(\log n)$ rounds.

Processes may still have same identifiers as those contained in fake tuples.  By a similar argument to the above, we can see additional $O(\log \ell)$ expected rounds are sufficient to give each process an identifier distinct from any fake one.  Then, all the fake identifiers are removed from {\it Topology}$_i$ of each process $i$ since such identifiers either become unreachable to $i$ in {\it Topology}$_i$ or become prefixes of real indentifiers.

After all identifiers become distinct from one another, the topology stored in {\it Topology}$_i$ of each process $i$ becomes stable if the process identifiers are ignored (\emph{i.e.}, only process identifiers get longer and longer).  
On the other hand, once the identifier of a process $u$ becomes lexicographically larger than that of a process $v$, $u$'s identifier is lexicographically larger than $v$'s afterward.  This guarantees that every execution of {\it UMIS}({\it Topology}$_i)$ at process $i$ after some point returns the same result concerning whether process $i$ is a member of the UMIS or not.  
By similar discussion to the proof of Theorem \ref{th:top-det} we can show that additional $O(D)$ rounds are sufficient to get the stable UMIS solution once every process has a unique identifier.  

Consequently, Algorithm \ref{alg:proba1} presents a self-stabilizing probabilistic UMIS algorithm and its expected convergence time is $O(\log n + \log \ell + D)$ rounds.
\qed

\subsection{Probabilistic solution with bounded memory in synchronous anonymous networks}
\label{sec:proba2}

The algorithm in the previous section is based on \emph{global} unique naming, however, self-stabilizing global unique naming in unidirectional networks inherently requires \emph{unbounded} memory.  The goal of this subsection is to achieve, with \emph{bounded} memory, a \emph{local} unique naming that gives each process an identifier that is different from that of any of its ancestors, and to compute a UMIS based on the previously computed local naming.  
Indeed, such a local naming is sufficient for each process to recognize the strongly connected component it belongs to.  Once the component is recognized, a UMIS can be computed by a method similar to that of the deterministic algorithm presented in Section \ref{sec:deter}.

In our scheme to achieve local unique naming, each process extends its identifier by appending a random bit when it finds an ancestor with the same identifier as its own. To be able to perform such a detection, a process needs to distinguish any of its ancestors from itself even when they have the same identifier. The detection mechanism is basically executed as follows: each process draws a random number, and disseminates its identifier together with the random number to its descendants.  When process $i$ receives the same identifier as its own, it checks whether the attached random number is same as its own. If they are different, the process detects that this is a distinct process (that is, a real ancestor) with the same identifier as its own current identifier. When the process receives the same identifier with the same random number as its own for a given period of time, it draws a new random number and repeats the above procedure. Hence, as two different processes eventually draw different random numbers, eventually every process is able to detect an ancestor with the same identifier if such an ancestor exists.

The above method may cause \emph{false detection} (or false positive) when a process receives its own identifier but with an old random number. To avoid such false detection, each identifier is relayed with a distance counter and is removed when the counter becomes sufficiently large.  Moreover, the process repeats the detection checks while keeping sufficiently long periods of time between them.
The details of the self-stabilizing probabilistic algorithm for the local naming are presented in Algorithm \ref{alg:naming}.

\begin{Algorithm}[htb]
\begin{tabbing}
xxx \= xxx \= xxx \= xxx \= xxx \= \kill
{\tt variables of process} $i$ \\
\> $id_i$: identifier (binary string) of $i$; \\
\> $rnd_i$: random number selected from $\{1, 2, \ldots , k \}$; 
   // $k\ (\ge 2)$ is a constant  \\
\> {\it ID}$_i$: set of $(id, rnd, dist)$ tuples; 
   // identifiers that $i$ is currently aware of. \\
\> \> // $id$: a process identifier \\
\> \> // $rnd$: random number of $P.(id)$ \\
\> \> // $dist$: distance that $id$ traverses \\
{\tt function} \\
\> {\it update(ID}$_i$)\\
\> \> {\it ID}$_i$ := $\{(id_i, rnd_i, 0)\} \cup 
      \bigcup_{j \in P.i}\{(id, rnd, dist+1)~|~(id, rnd,dist) \in 
      ${\it ID}$_j\}$; \\
\> \> {\tt while} $\exists (id, rnd, dist) \in$ {\it ID}$_i$ s.t. 
      $dist > |\{id~|~(id, *, *) \in ${\it ID}$_i\}|$; \\
\> \> \> remove $(id, rnd, dist)$ from {\it ID}$_i$; \\
\> \> {\tt if} ${\it timer} > |\{id~|~(id, *, *) \in ${\it ID}$_i\}|$ 
      // timer is incremented by one every round \\
\> \> \> ${\it naming(ID}_i$)\\
\> {\it naming(ID}$_i$)\\
\> \> {\tt if} $\exists (id_i, rnd, *) \in$ {\it ID}$_i$ s.t. $rnd \ne rnd_i$ \\
\> \> \> $id_i :=$ {\it append}($id_i$, {\it random\_bit}); 
      // append a random bit to the current id \\
\> \> $rnd_i :=$ number randomly selected from $\{1, 2, \ldots , k \}$; \\
\> \> {\it reset\_timer}; // reset timer to 0 \\
\> \> {\it update(ID}$_i$);\\
{\tt actions of process} $i$ \\
\> ${\tt true} \longrightarrow$ {\it update(ID}$_i$);
\end{tabbing}
\caption{Probabilistic local naming in synchronous anonymous networks}
\label{alg:naming}
\end{Algorithm}

\begin{lemma}
\label{lem:naming}
Algorithm \ref{alg:naming} presents a self-stabilizing probabilistic local naming algorithm in synchronous anonymous networks.  Its expected convergence time is $O((n + \ell) \log n)$ rounds where $\ell$ is the number of fake identifiers in the initial configuration.
\end{lemma}

\noindent
{\bf Proof sketch:}
First we show that the algorithm is a self-stabilizing probabilistic local naming algorithm.  For contradiction, we assume that two processes $i$ and $j$ (where $j$ is an ancestor of $i$) keep a same identifier after a configuration.  Without loss of generality, the distance from $j$ to $i$ is minimum among process pairs keeping same identifiers.  Let $j, u_1, u_2, \ldots , u_m, i$ be the shortest path from $j$ to $i$.  Since all processes in the path have mutually distinct identifiers except for a pair $i$ and $j$, $(id_j, rnd_j)$ is not discarded in the intermediate processes and is delivered to $i$.  Thus, eventually $i$ detects $id_i=id_j$ and $rnd_i \ne rnd_j$.  Then $i$ extends its identifier by adding a random bit, which is a contradiction.

We evaluate the expected convergence time of the algorithm.  By similar argument to the proof of Theorem \ref{th:proba1}, we can show that the expected number of bits added to a process identifier is $O(\log n)$.  Notice that the number $\ell$ of fake identifiers has no influence to the evaluation, for the distance {\it dist} of a fake identifier is larger than the timer value (once the timer is reset) and thus is removed (because of ${\it dist} > |\{ {\it id}~|~({\it id},*,*)\in {\it ID}_i\}|$) when function {\it naming} is executed.  On the other hand, in the scenario where all processes start with a same identifier, the time between two executions of function {\it naming} at a process is $O(n + \ell)$.  Thus, the expected convergence time is $O((n + \ell) \log n)$ rounds.
\qed

\begin{Algorithm}[htb]
\begin{tabbing}
xxx \= xxx \= xxx \= xxx \= xxx \= \kill
{\tt constants of process} $i$ \\
\> $id_i$: identifier of $i$; // distinct from that of any ancestor \\
\> $P_i$: identifier set of $P.i$; \\
{\tt variables of process} $i$ \\
\> {\it umis}$_i$: boolean; // {\tt true} iff $i$ is a UMIS node \\
\> {\it Topology}$_i$: set of ({\it id, ID}) tuples; 
   // topology that $i$ is currently aware of. \\
\> \> // {\it id}: a process identifier \\
\> \> // {\it ID}: identifier set of $P.(id)$ \\
\> {\it Comp}$_i$: identifier set
   // of processes in the strongly-connected component of $i$ \\
{\tt function} \\
\> {\it update(Topology}$_i$)\\
\> \> {\it Topology}$_i := \{(id_i, P_i)\} \cup 
      \bigcup_{j \in P.i}${\it Topology}$_j$; \\
\> {\it UMIS(Topology}$_i$)\\
\> \> {\it Comp}$_i$ := \{{\it id}$~|~${\it id} is reachable from $i$ 
      in {\it Topology}$_i$\}; \\
\> \> \> // set of processes in the strongly connected component of $i$ \\
\> \> $UMIS_v := \emptyset$; \\
\> \> {\tt if} $\exists j \in P.i -$ {\it Comp}$_i$ s.t. 
      {\it umis}$_j =$ {\tt true} {\tt or}
      $\exists j \in$ {\it Comp}$_i$ s.t. ($j > i$ {\tt and} 
      {\it umis}$_j =$ {\tt true}) \{\\
\> \> \> {\it umis}$_i$ := {\it false}; {\tt output} {\tt false}; \\
\> \> \} \\
\> \> {\tt else} \{ \\
\> \> \> {\it umis}$_i$ := {\tt true}; {\tt output} {\tt true}; \\
\> \> \} \\
{\tt actions of process} $i$ \\
\> ${\tt true} \longrightarrow$ {\it update(Topology}$_i$);
   {\it UMIS(Topology}$_i);$ 
\end{tabbing}
\caption{UMIS algorithm in locally-named networks}
\label{alg:proba2}
\end{Algorithm}

Algorithm \ref{alg:proba2} presents a self-stabilizing UMIS algorithm in locally-named networks.  Thus, the \emph{fair composition}\cite{D00b} of the algorithm with the local-naming algorithm in Algorithm \ref{alg:naming} provides a self-stabilizing UMIS algorithm in synchronous anonymous networks.  For simplicity, we omit the code for removing fake initial information in Algorithm \ref{alg:proba2} since such fake initial information can be removed in a similar way to Algorithm \ref{alg:naming}.

\begin{lemma}
\label{lem:component2}
In the algorithm presented in Algorithm \ref{alg:proba2}, each process can exactly recognize the topology of the strongly connected component it belongs to in $O(D)$ rounds where $D$ is the diameter of the network.
\end{lemma}

\noindent
{\bf Proof sketch:}
It is obvious that variable {\it Topology}$_i$ of each process $i$ after $D$ rounds consists of tuples $(id, P.(id))$ from all the ancestors of $i$ .  Notice that the local naming allows two distinct processes to have a same identifier if they are mutually unreachable.  Thus, {\it Topology}$_i$ may contain a same tuple $(id, P)$ of two or more distinct processes and/or may contain two tuples $(id, P)$ and $(id, P')$ with a same $id$ but different predecessor sets $P$ and $P'$.

Each process constructs the following graph $G_i = (V_i, E_i)$: $V_i=\{id~|~(id, *)\in {\it Topology}_i\}$ and $E_i=\{(u,v)~|~(v, P)\in {\it Topology}_i$ s.t. $u \in P \}$.  In other words, $G_i$ can be obtained from the actual graph $G$ as follows: First consider the subgraph $G'_i$ induced by the ancestors of $i$ and $i$ itself, and then merge the processes with the same identifier into a single process.

It is obvious that all processes in $G_i$ are reachable to $i$.  What we have to show is that process $j$ is reachable from $i$ in $G_i$ (or $j$ belongs to the strongly connected component of $i$) if and only if $j$ is also reachable from $i$ in $G'_i$.  The if part is obvious since $G_i$ is obtained from $G'_i$ by merging processes.  The only if part holds as follows.  Consider two distinct processes $j$ and $j'$ with a same identifier if exist.  Since they are mutually unreachable but are reachable to $i$, they are unreachable from $i$ in $G'_i$ (otherwise one of them is reachable from the other). In construction of $G_i$ from $G'_i$, merging is applied only to processes \emph{unreachable from} $i$, that is, the merging has no influence on reachability \emph{from} $i$.  Thus, any process unreachable from $i$ in $G'_i$ remains unreachable from $i$ in $G_i$.
\qed

\begin{lemma}
\label{lem:proba2}
Algorithm presented in Algorithm \ref{alg:proba2} is a self-stabilizing (deterministic) UMIS algorithm in (asynchronous) locally-named networks.  Its convergence time is $O(n)$ rounds.
\end{lemma}

\noindent
{\bf Proof sketch:}
First from Lemma \ref{lem:component2}, every process correctly recognizes in $O(D)$ rounds all the processes in the same connected component.  Then consider a \emph{source} strongly connected component.   The process with the maximum identifier in the component becomes a \emph{stable} UMIS member.  After that the UMIS outputs of processes in the component become stable one by one in the descending order of identifiers.  It takes at most $O(n')$ rounds until all the processes in the component become stable, where $n'$ is the number of processes in the component.

The same argument can be applied to a source strongly connected component in the graph obtained from $G$ by removing the components with stabilized UMIS outputs.  By repeating the argument, we can show that the UMIS outputs of all the processes become stable in $O(n)$ rounds.  It is clear that the processes with the UMIS outputs of {\bf true} form a UMIS.
\qed

From Lemmas \ref{lem:naming} and \ref{lem:proba2}, the following theorem holds.

\begin{theorem}
\label{th:proba2}
Fair composition of algorithms presented in Algorithm \ref{alg:naming} and Algorithm \ref{alg:proba2} provides a self-stabilizing probabilistic UMIS algorithm in synchronous anonymous networks.  Its expected convergence time is $O((n + \ell) \log n)$ rounds where $\ell$ is the number of fake identifiers in the initial configuration. The space complexity of the resulting protocol is $O(n \log n)$.
\end{theorem}

\section{Conclusion}
\label{sec:conclusion}

Although in bidirectionnal networks, self-stabilizing maximal independent set is as difficult as vertex coloring~\cite{GT00c}, this work proves that in unidirectionnal networks, the computing power and memory that is required to solve the problem varies greatly. Silent solutions to unidirectional uniform networks coloring require $\Theta(\log n)$ (resp. $\Theta(\log \delta)$, where $\delta$ denotes the maximal degree of the communication graph) bits per process and have stabilization time $\Theta(n^2)$ (resp. $\Theta(1)$) when deterministic (resp. probabilistic) solutions are considered. By contrast, deterministic maximal independent set construction in uniform networks is \emph{impossible}, and silent maximal independent set construction is \emph{impossible}, regardless of the deterministic or probabilistic nature of the protocols.

While we presented positive results for the deterministic case with identifiers, and the non-silent probabilistic cases, there remains the immediate open question of the possibility to devise a probabilistic solution with bounded memory in asynchronous setting.

Another interesting issue for further research related to \emph{global} tasks. The global unique naming that we present in section~\ref{sec:proba1} solves a truly global problem in networks where global communication is \emph{not} feasible, by defining proper equivalences classes between various identifiers. The case of other classical global tasks in distributed systems (\emph{e.g.} leader election) is worth investigating.

\bibliographystyle{plain}
\bibliography{../../../biblio/biblio}

\end{document}